\documentclass[a4paper, 11pt, twoside, onecolumn]{article}
\usepackage[utf8]{inputenc}
\usepackage[top=1in, bottom=1.2in, left=1in, right=1in]{geometry}
\usepackage[hypertexnames=false]{hyperref}
\usepackage{amsmath, amssymb, amsthm}
\usepackage{enumerate}
\usepackage{graphicx}
\usepackage{array}
\usepackage{float}

\newcommand{\T}{{\rm T}}
\newcommand{\tr}{{\rm Tr\,}}

\newcommand{\cc}{\mathbb{C}}
\newcommand{\rr}{\mathbb{R}}

\newcommand{\gd}{\mathcal{D}}

\newcommand{\nn}{\mathcal{N}}

\newcommand{\bx}{\boldsymbol{x}}
\newcommand{\by}{\boldsymbol{y}}

\newcommand{\onb}{orthonormal basis }

\newcommand{\bra}[1]{\left\langle #1 \right|}
\newcommand{\ket}[1]{\left| #1 \right\rangle}

\newtheorem{proposition}{Proposition}[section]
\newtheorem{corollary}[proposition]{Corollary}
\newtheorem{theorem}[proposition]{Theorem}
\newtheorem{lemma}[proposition]{Lemma}

\newtheorem*{defn*}{Definition}
\newtheorem*{thm*}{Theorem}

\theoremstyle{definition}

\numberwithin{equation}{section}

\title{{\bf Comparing Geometric Discord and Negativity for\\ Bipartite States}}
\author{Priyabrata Bag, Santanu Dey, and Hiroyuki Osaka}
\date{}

\begin{document}
\maketitle

\begin{abstract}
\noindent The geometric discord $\mathcal{D}$ of a state is a measure of the
quantumness of the state and the negativity $\mathcal{N}$ is a measure of
the entanglement of a state. It was proved by D. Girolami and G. Adesso that
for states on $\mathbb{C}^2\otimes\mathbb{C}^2$, the geometric discord is
always greater than or equal to the square of the negativity and conjectured
that this holds in general. S. Rana and P. Parashar showed that this
relation does not hold for all states on $\mathbb{C}^2\otimes\mathbb{C}^n$
for $n>2$. We provide several analytic families of states on
$\mathbb{C}^2\otimes\mathbb{C}^3$ violating this relation. Certain upper and
lower bounds for $\mathcal{N}^2-\mathcal{D}$ are obtained for states on
$\mathbb{C}^m\otimes\mathbb{C}^n$ for any $m, n\in\mathbb{N}$.
\\
\mbox{}
\end{abstract}

\section{Introduction}
\label{intro}

The entanglement of a state is related to quantum correlations. From recent findings,
it has turned out that the notion of quantum discord is a good measure of quantum
correlations. Let $\Omega_0$ be the set of \emph{classical-quantum states}, given by
$\sum p_k\ket{\psi_k}\bra{\psi_k}\otimes\rho_k$. We denote $M_m\otimes M_n$ by
$m\otimes n$, where $m, n\in\mathbb{N}$.

For an $m\otimes n$ ($m\leqslant n$) state $\rho$, the \emph{geometric discord (GD)}
is a variant of quantum discord which is defined as
\begin{align} \label{dis}
\mathcal{D}(\rho)&=\frac{m}{m-1}\min_{\chi\in\Omega_0}\|\rho-\chi\|^2 \notag\\
&=\frac{m}{m-1}\min_{\Pi^A}\|\rho-\Pi^A(\rho)\|^2,
\end{align}
where $\|C\|^2=\tr(C^{\dagger}C)$ is the Hilbert-Schmidt norm for any matrix $C$, and
in the last equality (cf. \cite{LF10:GMQD}), the minimization is over all possible
von Neumann measurements $\Pi^A=\{\Pi_k^A\}$ (that is, a set consisting of
one-dimensional orthogonal projectors summing to the identity) on $\rho^A$ and
$\Pi^A(\rho):=\sum_k(\Pi_k^A\otimes I^B)\rho(\Pi_k^A\otimes I^B)$. Here, $A$ and $B$
have been used to indicate the first part $\mathbb{C}^m$ and the second part
$\mathbb{C}^n$, respectively, of the system $\mathbb{C}^m\otimes\mathbb{C}^n$, and
$\rho^A$ is the marginal of $\rho$ on $\mathbb{C}^m$.
Let $\boldsymbol{\zeta}=(\zeta_1,\zeta_2,\ldots,\zeta_{d^2-1})^{\T}$ and $\zeta_i$ be
the generators of ${\rm SU}(d)$ for dimension $d=m$ or $n$.
Suppose $\rho$ have the Bloch form
\begin{equation}
\rho=\frac{1}{mn}\left[I_m\otimes I_n+\bx^{\T}\boldsymbol{\zeta}\otimes I_n+
I_m\otimes\by^{\T}\boldsymbol{\zeta}+\sum T_{ij}\zeta_i\otimes\zeta_j\right].
\label{blf}
\end{equation}
The following inequality, derived in \cite{RP12:TLBGD},
\begin{equation}
\mathcal{D}(\rho)\geqslant\frac{2}{m(m-1)n}\left[\|\bx\|^2+\frac{2}{n}\|T\|^2-
\sum_{k=1}^{m-1}\lambda_k^{\downarrow}(G)\right] \label{bgd}
\end{equation}
gives a tight lower bound on GD, where $\lambda_k^{\downarrow}(G)$ denotes the
eigenvalues of $G:=\bx\bx^{\T}+\dfrac{2}{n}TT^{\T}$ sorted in nonincreasing order.
The equality holds in \eqref{bgd} for all $2\otimes n$ states
(cf. \cite{RP12:TLBGD}, \cite{VR12:QDQQS}).

To calculate GD of a given $2\otimes 3$ state, we use the formula given by the
equality in \eqref{bgd}. For this we first find the Bloch form \eqref{blf} for the
given state by fixing the generators of ${\rm SU}(2)$ as the Pauli matrices, namely,
\[\sigma_1=\begin{pmatrix}
0 & 1\\
1 & 0
\end{pmatrix}, \,\,
\sigma_2=\begin{pmatrix}
0 & -\iota\\
\iota & 0
\end{pmatrix}, \,\,
\sigma_3=\begin{pmatrix}
1 & 0\\
0 & -1
\end{pmatrix},\]
and the generators of ${\rm SU}(3)$ as the Gell-Mann matrices, namely,
\[\mu_1=\begin{pmatrix}
0 & 1 & 0\\
1 & 0 & 0\\
0 & 0 & 0
\end{pmatrix}, \,\,
\mu_2=\begin{pmatrix}
0 & -\iota & 0\\
\iota & 0 & 0\\
0 & 0 & 0
\end{pmatrix}, \,\,
\mu_3=\begin{pmatrix}
1 & 0 & 0\\
0 & -1 & 0\\
0 & 0 & 0
\end{pmatrix},\]
\[\mu_4=\begin{pmatrix}
0 & 0 & 1\\
0 & 0 & 0\\
1 & 0 & 0
\end{pmatrix}, \,\,
\mu_5=\begin{pmatrix}
0 & 0 & -\iota\\
0 & 0 & 0\\
\iota & 0 & 0
\end{pmatrix}, \,\,
\mu_6=\begin{pmatrix}
0 & 0 & 0\\
0 & 0 & 1\\
0 & 1 & 0
\end{pmatrix},\]
\[\mu_7=\begin{pmatrix}
0 & 0 & 0\\
0 & 0 & -\iota\\
0 & \iota & 0
\end{pmatrix}, \,\,
\mu_8=\frac{1}{\sqrt{3}}\begin{pmatrix}
1 & 0 & 0\\
0 & 1 & 0\\
0 & 0 & -2
\end{pmatrix}.\]

Denote the linear partial transpose on $M_m\otimes M_n$ by
$\Gamma=\T\otimes {\rm id}_n$, that is, it maps $\ket{i}\bra{j}\otimes\ket{k}\bra{l}$
to $\ket{j}\bra{i}\otimes\ket{k}\bra{l}$.
A popular measure for the  entanglement of a state $\rho$ is the \emph{negativity}
(cf. \cite{VW02:CME}) of $\rho$, which is defined as
\begin{equation}
\nn(\rho)=\frac{\|\rho^{\Gamma}\|_1-1}{m-1}=\frac{2}{m-1}
\sum_{\lambda_i<0}|\lambda_i(\rho^{\Gamma})|, \label{neg}
\end{equation}
where  $\|X\|_1$ denotes the trace norm given by $\|X\|_1=\tr|X|$.

For $2\otimes 2$ states, it is shown in \cite{GA11:ICME}, that
$\mathcal{D}\geqslant\nn^2$. S. Rana and P. Parashar discussed $2\otimes n$ case in
\cite{RP12:ELBGD}, and justified that this $\mathcal{D}\geqslant\nn^2$ does not hold
for $n>2$. But they also deduced that the occurrence of the violation of this
inequality is very rare for $n=3$. The example of the state provided by them for
$n=3$ for which $\mathcal{D}<\nn^2$ involves long decimals and complex entries. The
quantity $\nn^2-\mathcal{D}$ is very small in their example.
We find several parametrized families of $2\otimes 3$ states when $\mathcal{D}<\nn^2$
in Section~\ref{GDN}. We have examples where the entries of the matrices are simple
fractions or integers and $\nn^2-\mathcal{D}$ is not so small.

S. Rana in \cite{Ran13:NEPT} proved that for any $m\otimes n$ state $\rho$, the
number $(m-1)(n-1)$ is an upper bound for the number of negative eigenvalues of
$\rho^{\Gamma}$. Later N. Johnston in \cite{Jon13:NPPTS} showed that for all $m$ and
$n$, there exists an $m\otimes n$ state $\rho$ such that $\rho^{\Gamma}$ has
$(m-1)(n-1)$ negative eigenvalues. All parametrized families of $2\otimes 3$ states
$\rho$ in Section~\ref{GDN} for which $\mathcal{D}<\nn^2$ are such that
$\rho^{\Gamma}$ have the maximum possible negative eigenvalues.

In Section~\ref{BGDN}, we show that for an $m\otimes n$ state $\rho$
\begin{equation} \label{bds}
-\frac{m}{m-1}\leqslant\nn(\rho)^2-\gd(\rho)\leqslant 1.
\end{equation}
This makes use of convexity and some other property of negativity $\nn$ of states
that was established by G. Vidal and R. F. Werner in \cite{VW02:CME}.

\section{Examples of \texorpdfstring{$2\otimes 3$}{2x3} States with
\texorpdfstring{$\mathcal{D}<\nn^2$}{D<N2}}
\label{GDN}

For $a,b\in\mathbb{R}$ with $b>0$, consider the following element in
$M_2\otimes M_3$:
\begin{equation}
\rho_1(a,b)=\frac{1}{2(a^2+b^2)}\begin{pmatrix}
a^2 & 0 & 0 & 0 & ab & 0\\
0 & b^2 & 0 & 0 & 0 & ab\\
0 & 0 & 0 & 0 & 0 & 0\\
0 & 0 & 0 & 0 & 0 & 0\\
ab & 0 & 0 & 0 & b^2 & 0\\
0 & ab & 0 & 0 & 0 & a^2
\end{pmatrix}. \label{rho1ab}
\end{equation}
The eigenvalues of $\rho_1(a,b)$ are $0$ with multiplicity $4$ and $\dfrac{1}{2}$
with multiplicity $2$. Since $\rho_1(a,b)$ is Hermitian, this guarantees that
$\rho_1(a,b)$ is a state. The eigenvalues of $\rho_1(a,b)^{\Gamma}$ are
$\dfrac{a^2}{2(a^2+b^2)}$, $\dfrac{b^2+b\sqrt{b^2+4a^2}}{4(a^2+b^2)}$ and
$\dfrac{b^2-b\sqrt{b^2+4a^2}}{4(a^2+b^2)}$, all with multiplicity $2$. Clearly,
it has two negative eigenvalues for $a\neq 0$, namely,
$\dfrac{b^2-b\sqrt{b^2+4a^2}}{4(a^2+b^2)}$, with multiplicity $2$. This is the
maximum number of negative eigenvalues possible for the partial transpose of any
$2\otimes 3$ state.
The negativity of $\rho_1(a,b)$ can be calculated from the second equality of
\eqref{neg} as 
\begin{equation} \label{nr1}
\nn(\rho_1(a,b))=\dfrac{b\sqrt{b^2+4a^2}-b^2}{a^2+b^2}
\end{equation}
for $a\neq 0$. Observe that $\nn(\rho_1(a,b))=0$ when $a=0$, and this coincides with
the expression for $\nn(\rho_1(a,b))$ in $\eqref{nr1}$. Therefore
$\nn(\rho_1(a,b))^2=\dfrac{2b^4+4a^2b^2-2b^3\sqrt{b^2+4a^2}}{(a^2+b^2)^2}$.

For the state $\rho_1(a,b)$ in \eqref{rho1ab}, the Bloch form is determined by
$\bx^{\T}=(0,0,0)$, $\by^{\T}=\left(0,0,\dfrac{3a^2-6b^2}{4(a^2+b^2)},0,0,0,0,
-\dfrac{(a^2-2b^2)\sqrt{3}}{4(a^2+b^2)}\right)$ and
\begin{equation*}
T=\begin{pmatrix}
\dfrac{3ab}{2(a^2+b^2)} & 0 & 0 & 0 & 0 & \dfrac{3ab}{2(a^2+b^2)} & 0 & 0\\
0 & -\dfrac{3ab}{2(a^2+b^2)} & 0 & 0 & 0 & 0 & -\dfrac{3ab}{2(a^2+b^2)} & 0\\
0 & 0 & \dfrac{3a^2}{4(a^2+b^2)} & 0 & 0 & 0 & 0 & \dfrac{3\sqrt{3}a^2}{4(a^2+b^2)}
\end{pmatrix}.
\end{equation*}
Thus,
\begin{equation*}
TT^{\T}=\begin{pmatrix}
\dfrac{9a^2b^2}{2(a^2+b^2)^2} & 0 & 0\\
0 & \dfrac{9a^2b^2}{2(a^2+b^2)^2} & 0\\
0 & 0 & \dfrac{9a^4}{4(a^2+b^2)^2}
\end{pmatrix},
\end{equation*}
and hence
\begin{equation*}
G=\begin{pmatrix}
\dfrac{3a^2b^2}{(a^2+b^2)^2} & 0 & 0\\
0 & \dfrac{3a^2b^2}{(a^2+b^2)^2} & 0\\
0 & 0 & \dfrac{3a^4}{2(a^2+b^2)^2}
\end{pmatrix}.
\end{equation*}
If $a^2\geqslant 2b^2$, then the eigenvalues of $G$ in nonincreasing order are
$\dfrac{3a^4}{2(a^2+b^2)^2},\dfrac{3a^2b^2}{(a^2+b^2)^2}$
and $\dfrac{3a^2b^2}{(a^2+b^2)^2}$. Also,
$\|T\|^2=\tr(T^{\dagger}T)=\tr(TT^{\dagger})=\tr(TT^{\T})=
\dfrac{36a^2b^2+9a^4}{4(a^2+b^2)^2}$. Thus, the GD of $\rho_1(a,b)$ is
$\mathcal{D}(\rho_1(a,b))=\dfrac{2a^2b^2}{(a^2+b^2)^2}$,
when $a^2\geqslant 2b^2$. Similarly, it can be shown that $\mathcal{D}(\rho_1(a,b))=
\dfrac{a^4+2a^2b^2}{2(a^2+b^2)^2}$, when $a^2\leqslant 2b^2$.
Therefore, for $a^2>2b^2$, $\nn(\rho_1(a,b))^2-\mathcal{D}(\rho_1(a,b))=
\dfrac{2b^4+2a^2b^2-2b^3\sqrt{b^2+4a^2}}{(a^2+b^2)^2}>0$.
Thus, for $a^2>2b^2$ with $b>0$, the state $\rho_1(a,b),$
violates the inequality $\mathcal{D}\geqslant\nn^2$. In particular, for $a=5$ and
$b=2$, we obtain the following state:
\begin{equation*}
\rho_1(5,2)=\frac{1}{58}\begin{pmatrix}
25 & 0 & 0 & 0 & 10 & 0\\
0 & 4 & 0 & 0 & 0 & 10\\
0 & 0 & 0 & 0 & 0 & 0\\
0 & 0 & 0 & 0 & 0 & 0\\
10 & 0 & 0 & 0 & 4 & 0\\
0 & 10 & 0 & 0 & 0 & 25
\end{pmatrix}
\end{equation*}
which violates the inequality $\gd\geqslant\nn^2$. Here, $\nn(\rho_1(5,2))^2-
\mathcal{D}(\rho_1(5,2))=\dfrac{232-32\sqrt{26}}{841}\approx 0.0818$.
Let $c$ denote $\dfrac{a}{b}$. Then, $\nn(\rho_1(a,b))^2$ takes the form
$\dfrac{4c^2+2-2\sqrt{4c^2+1}}{(c^2+1)^2}$ for all $c\in\rr$, and
\[\gd(\rho_1(a,b))=\begin{cases}
\dfrac{2c^2}{(c^2+1)^2} & \mbox{when } c^2\geqslant 2,\\
\dfrac{c^4+2c^2}{2(c^2+1)^2} & \mbox{when } c^2\leqslant 2.
\end{cases}
\]
Using this representation, we obtain the following figure:
\begin{figure}[H]
	\begin{center}
		\includegraphics[width=4.25in]{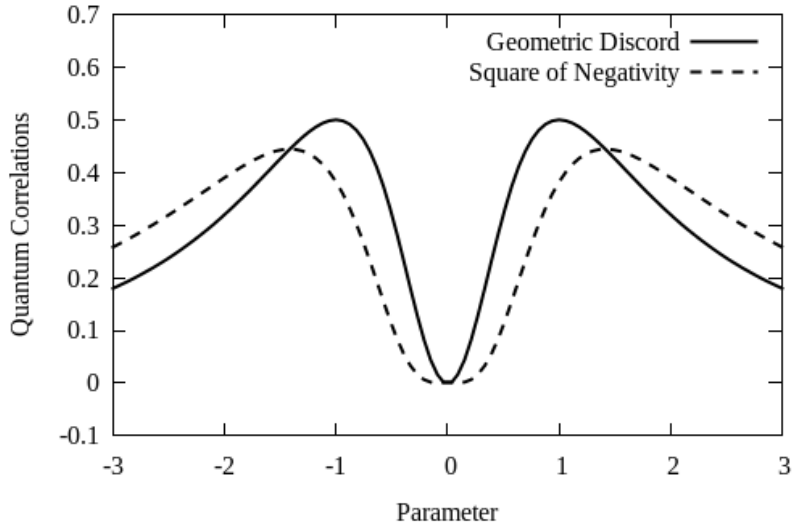}
	\end{center}
	\caption{Geometric Discord and Square of Negativity for $\rho_1(a,b)$.}
\end{figure}

Similarly, for the following three families of $2 \otimes 3$ states, it can be shown
that the inequality $\mathcal{D}\geqslant\nn^2$ is violated:
The first family is given by
\begin{equation*}
\rho_2(a)=\frac{1}{8a+2}\begin{pmatrix}
3a+1 & 0 & 0 & 0 & 2a & 0\\
0 & a & 0 & 0 & 0 & 2a\\
0 & 0 & 0 & 0 & 0 & 0\\
0 & 0 & 0 & 0 & 0 & 0\\
2a & 0 & 0 & 0 & a & 0\\
0 & 2a & 0 & 0 & 0 & 3a+1
\end{pmatrix}, 
\quad\mbox{where $0<a\leqslant 1$,}
\end{equation*}
for which the graphs of $\gd$ and $\nn^2$ are illustrated in the following figure:
\begin{figure}[H]
	\begin{center}
		\includegraphics[width=4.25in]{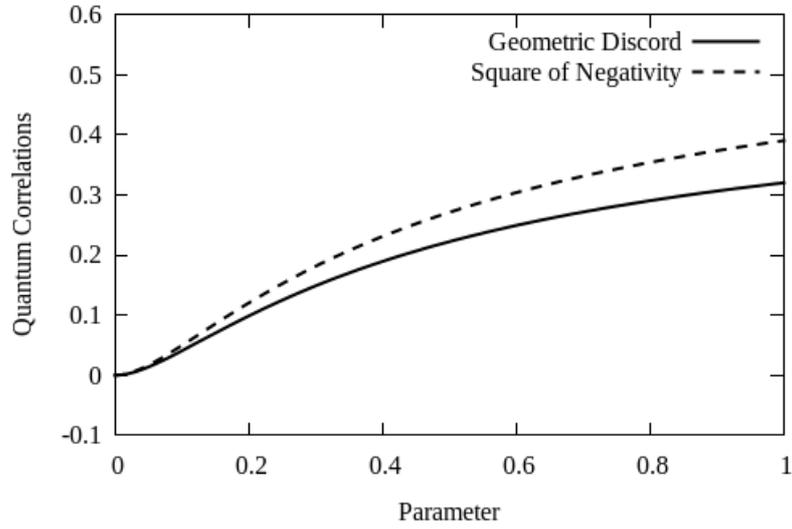}
	\end{center}
	\caption{Geometric Discord and Square of Negativity for $\rho_2(a)$.}
\end{figure}
The second family is given by
\begin{equation*}
\rho_3(a)=\frac{1}{8a+2}\begin{pmatrix}
3a+1 & 0 & 0 & 0 & 2a-1 & 0\\
0 & a & 0 & 0 & 0 & 2a-1\\
0 & 0 & 0 & 0 & 0 & 0\\
0 & 0 & 0 & 0 & 0 & 0\\
2a-1 & 0 & 0 & 0 & a & 0\\
0 & 2a-1 & 0 & 0 & 0 & 3a+1
\end{pmatrix}, 
\quad\mbox{where $\dfrac{7}{4}\leqslant a\leqslant \dfrac{19}{4}$,}
\end{equation*}
for which the graphs of $\gd$ and $\nn^2$ are illustrated in the following figure:
\begin{figure}[H]
	\begin{center}
		\includegraphics[width=4.25in]{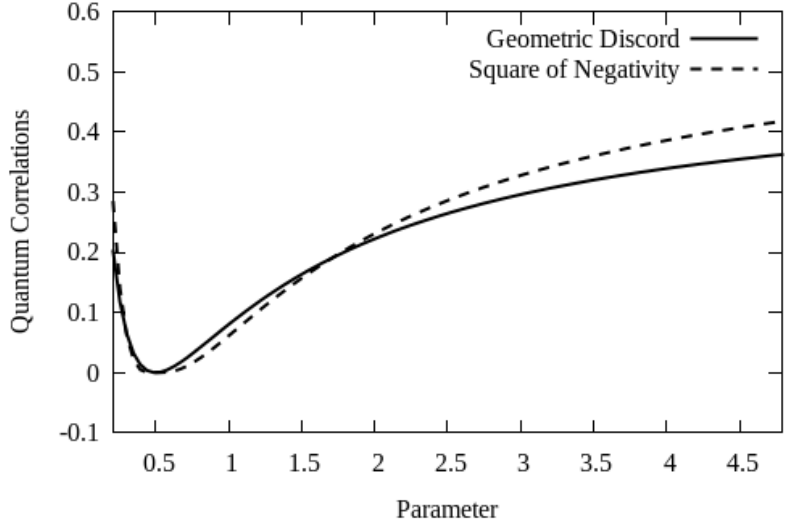}
	\end{center}
	\caption{Geometric Discord and Square of Negativity for $\rho_3(a)$.}
\end{figure}
Finally, the third family is given by
\begin{equation*}
\rho_4(a)=\frac{1}{8a+2}\begin{pmatrix}
3a+1 & 0 & 0 & 0 & 2a-2 & 0\\
0 & a & 0 & 0 & 0 & 2a-2\\
0 & 0 & 0 & 0 & 0 & 0\\
0 & 0 & 0 & 0 & 0 & 0\\
2a-2 & 0 & 0 & 0 & a & 0\\
0 & 2a-2 & 0 & 0 & 0 & 3a+1
\end{pmatrix}, 
\quad\mbox{where $\dfrac{7}{2}\leqslant a\leqslant \dfrac{17}{2}$,}
\end{equation*}
for which the graphs of $\gd$ and $\nn^2$ are illustrated in the following figure:
\begin{figure}[H]
	\begin{center}
		\includegraphics[width=4.25in]{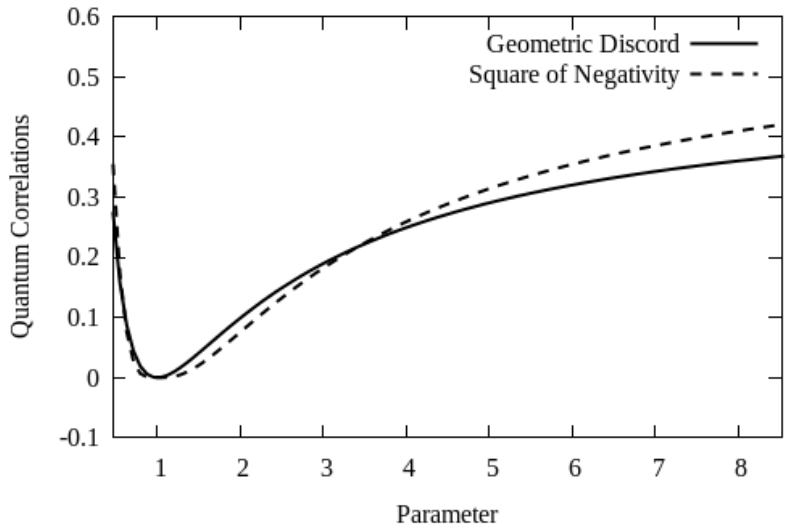}
	\end{center}
	\caption{Geometric Discord and Square of Negativity for $\rho_4(a)$.}
\end{figure}

\section{Bounds for \texorpdfstring{$\nn^2-\gd$}{N2-D}}
\label{BGDN}

In this section, we first obtain bounds for $\nn$ and $\gd$ individually of an
$m\otimes n$ state which together give bounds for $\nn^2-\gd$. We assume throughout
the section that $m\leqslant n$. From the definition, it follows that both of $\nn$
and $\gd$ are nonnegative, and which give obvious lower bounds $0$ for both of them.
An upper bound for $\gd$ is obtained in the next proposition.

\begin{proposition} \label{ubfgd}
For any $m\otimes n$ state $\rho$,
\[0\leqslant\gd(\rho)\leqslant\frac{m}{m-1}.\]
\end{proposition}
\begin{proof}
Here, we use the second expression for $\gd$ in \eqref{dis} to find the upper bound.
Let $\Pi^A$ be a von Neumann measurement on $\rho^{A}$, that is, there is an \onb
$\{\ket{\psi_k} : 1 \leqslant k \leqslant m\}$ of $\cc^m$ such that
$\Pi^A(\rho)=\sum\limits_{k=1}^m (\ket{\psi_k}\bra{\psi_k}\otimes I_n)
\rho(\ket{\psi_k}\bra{\psi_k}\otimes I_n)$.
For any von Neumann measurement $\Pi^A$, we have the following:
\begin{align*}
&\tr((\Pi^A(\rho))^2)\\
=&\tr\left(\sum_{k=1}^m(\ket{\psi_k}\bra{\psi_k}\otimes I_n)\rho(\ket{\psi_k}
\bra{\psi_k}\otimes I_n)\sum_{l=1}^m(\ket{\psi_l}\bra{\psi_l}\otimes I_n)
\rho(\ket{\psi_l}\bra{\psi_l}\otimes I_n)\right)\\
=&\sum_{k=1}^m\tr((\ket{\psi_k}\bra{\psi_k}\otimes I_n)\rho(\ket{\psi_k}\bra{\psi_k}
\otimes I_n)\rho(\ket{\psi_k}\bra{\psi_k}\otimes I_n))\\
=&\sum_{k=1}^m\tr(\rho(\ket{\psi_k}\bra{\psi_k}\otimes I_n)
\rho(\ket{\psi_k}\bra{\psi_k}\otimes I_n))\\
=&\tr(\rho\Pi^A(\rho)).
\end{align*}
Thus,
\begin{align*}
\|\rho - \Pi^A(\rho)\|^2 =& \tr(\rho^2) - 2\tr(\rho\Pi^A(\rho)) +
\tr((\Pi^A(\rho))^2)\\
=& \tr(\rho^2) - 2\tr((\Pi^A(\rho))^2) + \tr(\rho\Pi^A(\rho))\\
=& \tr(\rho^2) - \tr((\Pi^A(\rho))^2) \leqslant 1.
\end{align*}
Therefore $0 \leqslant \gd(\rho) \leqslant \dfrac{m}{m-1}$.
\end{proof}

For any $m\otimes n$ pure state $\rho$ it is proved in
\cite[Proposition~1]{LF12:EGMQD} that $\gd(\rho)\leqslant 1$.

\begin{lemma}
$\nn(\rho) \leqslant 1$ and $\gd(\rho) \leqslant 1$ for any $m \otimes n$
pure state $\rho$.
\end{lemma}
\begin{proof}
Since $\rho$ is a pure state, $\rho = |\Phi\rangle \langle\Phi|$ for some vector
$|\Phi\rangle$ with Schmidt decomposition form
\begin{equation} \label{sc}
|\Phi\rangle = \sum_{i=1}^lc_i |e_i\rangle \otimes |f_i\rangle
\end{equation}
for some orthonormal sets $\{|e_i\rangle : 1 \leqslant i \leqslant l\}$ and
$\{|f_i\rangle : 1 \leqslant i \leqslant l\}$, where $1 \leqslant l \leqslant m$ and
$c_i > 0$ for $1 \leqslant i \leqslant l$. From \cite[Proposition~8]{VW02:CME} we
obtain
\begin{equation} \label{tnn}
\widetilde{\nn}(\rho) = \frac{1}{2}\left[\left(\sum_{i=1}^lc_i\right)^2 -1\right],
\end{equation}
where $\widetilde{\nn}(\rho) = \dfrac{m-1}{2} \nn(\rho)$.
Since $\sum\limits_{i=1}^lc_i^2 = 1$, we obtain
\begin{equation*}
\sum_{i=1}^lc_i \leqslant \left(\sum_{i=1}^l1\right)^{\frac{1}{2}}
\left(\sum_{i=1}^lc_i^2\right)^{\frac{1}{2}} = l^{\frac{1}{2}}.
\end{equation*}
Therefore $\nn(\rho) \leqslant \dfrac{1}{m-1}(l - 1) \leqslant 1$.
It is proved in \cite[Proposition~1]{LF12:EGMQD} that
\begin{equation} \label{tgd}
\widetilde{\gd}(\rho) = 1 - \sum_{i=1}^l c^4_i,
\end{equation}
where $\widetilde{\gd}(\rho) = \dfrac{m-1}{m} \gd(\rho)$.
It follows that $\gd(\rho) \leqslant 1$.
\end{proof}

For $m, n \in \mathbb{N}$ with $2 \leqslant m \leqslant n$, let $\rho_{m,n,max} =
\displaystyle\dfrac{1}{m}\sum\limits_{i,j=1}^m e_{ij} \otimes f_{ij}$ in
$M_m\otimes M_n$, where $\{e_{ij}\}_{i,j=1}^m$ and $\{f_{ij}\}_{i,j=1}^n$ are
matrix units for $M_m$ and $M_n$, respectively. When $m = n$, the state
$\rho_{n,n,max}$ is a maximal entanglement.
Since for $\rho_{m,n,max}$ in the decomposition \eqref{sc} $l=m$ and $c_i =
\dfrac{1}{\sqrt{m}}$ for $1 \leqslant i \leqslant m$, it follows from \eqref{tnn}
and \eqref{tgd} that $\nn (\rho_{m,n,max}) = 1$ and $\gd(\rho_{m,n,max})=1$. Thus,
\[\nn(\rho_{m,n,max})^2 - \gd(\rho_{m,n,max})=0.\]

\begin{corollary} \label{cor:negativity}
If $\rho$ is an $m\otimes n$ state, then $\nn(\rho) \leqslant 1$.
\end{corollary}
\begin{proof}
Since any state can be written as a convex sum of pure states and the convexity
of $\nn$ by \cite[Proposition~1]{VW02:CME}, we have the conclusion.
\end{proof}

We conclude that the negativity $\nn$ achieves its maximum value on the states
$\rho_{m,n,max}$.

\begin{theorem}
For any $m\otimes n$ state $\rho$
\[-\frac{m}{m-1} \leqslant \nn(\rho)^2 - \gd(\rho) \leqslant 1.\]
\end{theorem}
\begin{proof}
The proof is immediate from Proposition~\ref{ubfgd} and
Corollary~\ref{cor:negativity}.
\end{proof}

Note that $|\nn(\rho)^2 - \gd(\rho)| \leqslant 1$ for any pure $m\otimes n$
state $\rho$.

\section*{Acknowledgement}
The third author is partially supported by the JSPS 
KAKENHI Grant Number JP17K05285.

\bibliography{CGDNeBSBib}
\bibliographystyle{plain}

\vspace*{0.5cm}
\noindent Priyabrata Bag\\
School of Mathematical Sciences,\\
SVKM's NMIMS (Deemed to be University),\\
V. L. Mehta Road, Vile Parle (West),\\
Mumbai, Maharashtra 400056, India.\\
E-mail: {\it priyabrata.bag@nmims.edu}

\vspace*{0.5cm}
\noindent Santanu Dey\\
Department of Mathematics,\\
Indian Institute of Technology Bombay,\\
Mumbai, Maharashtra 400076, India.\\
E-mail: {\it santanudey@iitb.ac.in}

\vspace*{0.5cm}
\noindent Hiroyuki Osaka\\
Department of Mathematical Sciences,\\
Ritsumeikan University,\\
Kusatsu, Shiga, 525-8577, Japan.\\
E-mail: {\it osaka@se.ritsumei.ac.jp}

\end{document}